\documentclass{article}

\usepackage[english]{babel}

\usepackage[letterpaper,top=2cm,bottom=2cm,left=3cm,right=3cm,marginparwidth=1.75cm]{geometry}

\usepackage{graphicx}
\usepackage[colorlinks=true, allcolors=blue]{hyperref}

\usepackage{authblk}
\usepackage{multirow}%
\usepackage{amsmath,amssymb,amsfonts}%
\usepackage{amsthm}%
\usepackage{mathrsfs}%
\usepackage{xcolor}%
\usepackage{textcomp}%
\usepackage{manyfoot}%
\usepackage{booktabs}%
\usepackage{url}

\usepackage{wrapfig}
\usepackage[export]{adjustbox} 
\usepackage{subfigure}
\usepackage{subcaption}
\usepackage{comment}
\usepackage{ulem}
\usepackage{ifthen}
\usepackage{lmodern} 
\usepackage{anyfontsize} 
\usepackage{xcolor}
\definecolor{teal}{rgb}{0.0, 0.5, 0.5}

\newtheorem{theorem}{Theorem}
\newtheorem{proposition}[theorem]{Proposition}%
\newtheorem{Lemma}{Lemma}
\newtheorem{assumption}{Assumption}%
\newcommand{\equaref}[1]{(\ref{#1})}
\def\ind{{\rm 1\hspace{-0.90ex}1}}

\title{Dominance or Fair Play in Social Networks? \\ A Model of Influencer Popularity Dynamics}
\author[1]{Franco Galante}
\author[2]{Chiara Ravazzi}
\author[1]{Luca Vassio}
\author[3]{Michele Garetto}
\author[1,2]{Emilio Leonardi}
\affil[1]{Politecnico di Torino, Italy}
\affil[2]{National Research Council (CNR-IEIIT), Italy}
\affil[3]{Università di Torino, Italy}
\date{}  

\begin{document}
\maketitle

\begin{abstract}
This paper presents a data-driven mean-field approach to model the popularity dynamics of users seeking public attention, i.e., influencers. We propose a novel analytical model that integrates individual activity patterns, expertise in producing viral content, exogenous events, and the platform's role in visibility enhancement, ultimately determining each influencer's success. We analytically derive sufficient conditions for system ergodicity, 
enabling predictions of popularity distributions. A sensitivity analysis explores various system configurations, highlighting conditions favoring either dominance or fair play among influencers. Our findings offer valuable insights into the potential evolution of social networks towards more equitable or biased influence ecosystems.
\end{abstract}

\section{Introduction}\label{sec:intro}

Online social networks (OSNs) crucially shape digital interaction, increasingly replacing traditional media. This shift has created a fertile ground for the rise of a new class of online users: influencers. These are individuals with a large following on OSNs, who typically seek to expand it. They are also central to marketing because of their ability to drive consumer behavior \cite{9689053}. Therefore, understanding how influence hierarchies form and consolidate over OSNs is of great interest. Given the growing role of platforms in shaping opinions, it is crucial to determine whether dominant users rise through their merits or skewed feedback mechanisms that concentrate users' preferences.

Despite significant interest and numerous social network analyses, a gap remains in understanding the quantitative dynamics of popularity evolution. Firstly, a universal definition of popularity in social contexts is lacking, with the literature offering varied interpretations tied to specific scientific contexts. Popularity metrics are often identified with visibility metrics~\cite{Bakshy2011}, such as follower count or interactions \cite{vassio2022mining}. However, the follower count is relatively static and its validity has been questioned~\cite{Cha2010} due to its insensitivity to interaction dynamics and rare decrease~\cite{Bakshy2011}. Another research direction uses a system perspective, interpreting popularity dynamically (see~\cite{castaldo:tel-04001597} and references therein). Among these models are epidemic~\cite{richier2014bio}, self-exciting (e.g., Hawkes~\cite{Crane2008}), and Bass diffusion models~\cite{bass1969new}. It is established that opinion changes result from collective interaction and mutual influence~\cite{friedkin1999social}, but quantifying these influences is difficult~\cite{Ravazzi_CSM}. Mean-field models address this by replacing pairwise interactions with an average interaction across the population~\cite{Collet2010}, where individual dynamics depend on the aggregate system.

By adopting a mean-field perspective and drawing inspiration from OSN data (Facebook in this case), we propose a novel, comprehensive model that captures the evolution of real-world influencers by accounting for the following factors: (i)  Collective attention, a limited resource influencers compete for~\cite{Weng2012}, tends to decrease over time without new stimuli or maintenance~\cite{LorenzSpreen2019}; (ii) Content creation patterns, traditionally modeled as a homogeneous Poisson process~\cite{BESSI2017459}, are demonstrably shaped by the influencer's popularity~\cite{10.1145/1557914.1557983}; (iii) Post success (and its impact on influencer popularity) is modeled as a random variable potentially dependent on acquired popularity~\cite{Muchnik2013}, user characteristics (competence, experience, attractiveness), and platform feedback~\cite{Jiang2019} (although undisclosed, engagement maximization algorithms are recognized to rely on past popularity and user preferences~\cite{Covington2016}); (iv) Exogenous, platform-independent events, can influence popularity~\cite{Hilgartner1988}.

Beyond introducing the model, our primary contribution is empirical evidence supporting the analytical treatment's core hypotheses. Due to the system's stochasticity, the dynamics do not deterministically converge to a single state. However, under specific conditions, we prove the system's state, described by a Markov process, is ergodic, possessing a long-run invariant distribution. This distribution, while not explicitly formulated, solves a system of partial differential equations. 
This result enables a probabilistic study of popularity emergence and a comparison of influencer distributions by quantifying their emergence probability and duration in privileged positions. 

\section{Preliminary Empirical Analysis}\label{sec:empirical_analysis}
Using real data traces, we seek to characterize the popularity of influential individuals and empirically identify the key elements that contribute to the rise or fall of their prominence over time.

We focus on the social network Facebook\footnote{A parallel analysis of Instagram data has been completed and will be presented in a forthcoming extension of this work.}, where influencers publish content (i.e., posts), share content that followers, as well as other platform users, can view, like, and comment on.
We restrict the analysis to influencers who post primarily on 
chess (as the FIDE score provides an index of competence) and other two popular topics: cars and science.
We created lists of such influencers by actively searching for them on the platforms through hashtags and keywords and by consulting public lists available online.\footnote{e.g., 
\url{https://hireinfluence.com/blog/top-science-influencers/}}
Note that when we use the term influencer, we do not only mean physical individuals but also magazines, organizations, or companies.
For each monitored influencer, we analyzed all the data related to the posts published between January 1, 2014, and May 15, 2024, using the CrowdTangle tool and its API.\footnote{CrowdTangle was a public insights tool from Meta available to researchers. 
\url{https://transparency.meta.com/en-gb/researchtools/other-datasets/crowdtangle/}}
Available data include the number of followers at the publishing time and the time-series of engagement metrics for each post.  
In total, we collected~$1,965,805$ posts from~$111$ influencers. 


The number of followers is a straightforward and widely adopted proxy in the literature \cite{Bakshy2011} for gauging an influencer's popularity at a given time; however, it doesn't fully represent their popularity at that specific point. For example, we observe that the number of followers typically exhibits a purely increasing trend.
Figure~\ref{fig:osn_evidence}~(a) illustrates this pattern, reporting the followers' evolution on Facebook for the influencers who mainly post about chess.
This phenomenon occurs because users have no real incentive to unfollow an influencer. 
Even when influencers stop posting for extended periods, they retain most of their followers. To illustrate this, we examine inactive influencers, that is, those who have not posted on the platform for more than~6 months. Figure~\ref{fig:osn_evidence}~(b) shows the follower's variation after an inactivity period. In most cases, the follower's difference is negligible, with losses (in red) and gains (in green) occurring about equally.

\begin{figure}[h!]
    \centering
    \begin{tabular}{cc}
\includegraphics[height=4cm, valign=b]{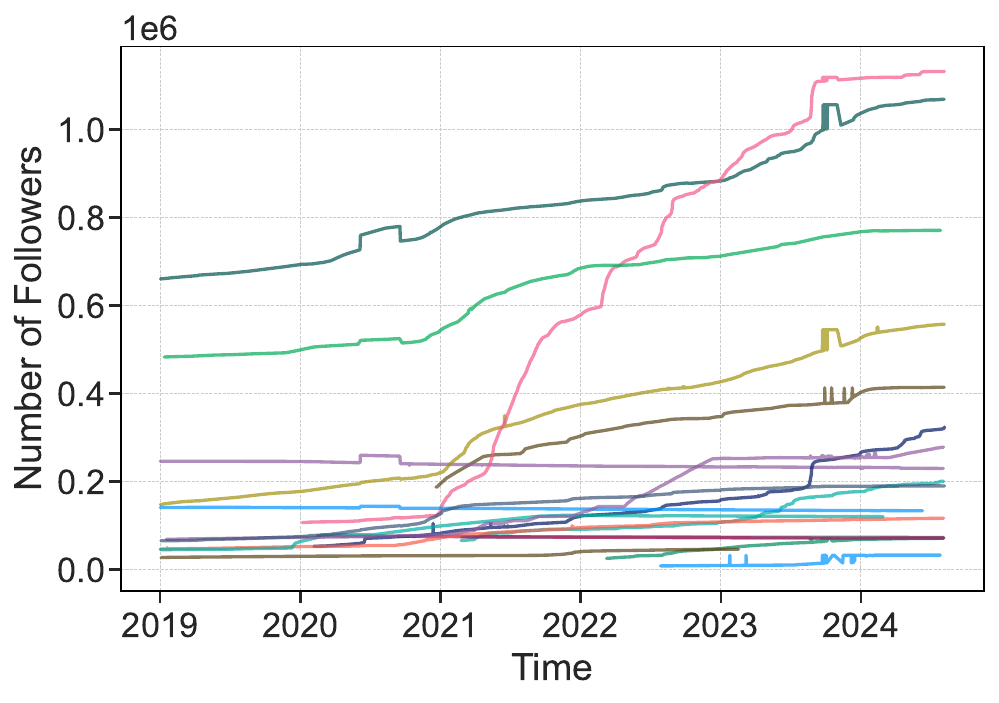}
& 
\hspace{1cm}
\includegraphics[height=3.8cm, valign=b]{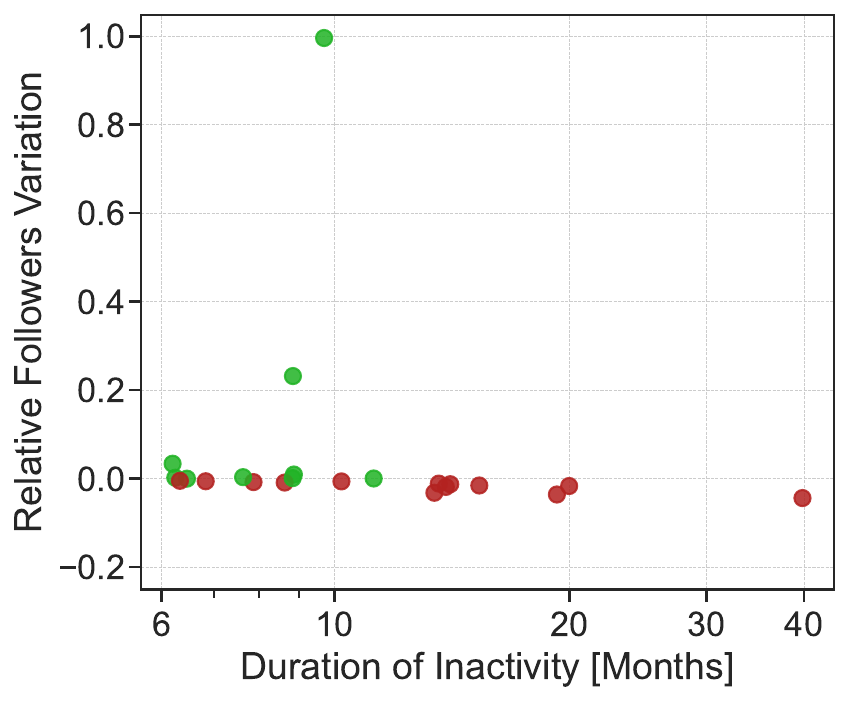} \\
{\scriptsize  $\quad$ (a) }& {\scriptsize $\quad\quad\quad\quad\quad$(b) } \\
\end{tabular}
     \caption{Temporal evolution of the number of followers for \textit{chess} influencers on Facebook (a) and a scatter plot of the relative followers' variation after a period of inactivity longer than 6~months (b), in log x-scale. Each line/dot is an influencer}.  \label{fig:osn_evidence}
\end{figure}

\section{A Mean-Field Model of Popularity}\label{sec:model}
To overcome the intrinsic limitations of the number of followers as popularity 
metric, we propose the following approach.

We consider a finite set of influencers $\mathcal{I}$ and assign to each influencer~$i \in \mathcal{I}$ a scalar variable~$X^{[i]}(t)\in\mathbb{R}$, representing their instantaneous \textit{effective popularity}. 
$X^{[i]}(t)$~is a measure of the level of interest among platform users towards the content posted by influencer~$i\in\mathcal{I}$ at a given time~$t\in\mathbb{R}$. 

Based on the empirical evidence of Section~\ref{sec:empirical_analysis}, we assume that~$X^{[i]}(t)$ is not simply proportional to the number of followers but is governed by autonomous dynamics, determined mainly by the following four factors:
(i) the natural tendency of users to lose or shift their interest;
(ii) the impact of competition from other influencers  on the same topic, diverting attention away from influencer~$i$;
(iii) the posting activity of influencer~$i$, aimed to secure current and attract new audiences; 
(iv) exogenous events affecting the online visibility of influencers, 
such as newsworthy events as reported by traditional media, which expose them to broader audiences beyond their established following.
Formally, we assume $X^{[i]}(t)$ evolves according to the following stochastic differential equation:
\begin{equation}\label{SDE}
	\mathrm{d} X^{[i]}(t)= - \gamma X^{[i]}(t) dt +V_i(t)N^{[i]}_I(dt) 
    + W_i(t) N_E^{[i]}(dt)
\end{equation}
where each of the terms on the right-hand side corresponds to some of the previously described factors.
More in detail, we first assume that~$X^{[i]}(t)$ decreases at a rate~$\gamma$ 
in the absence of additional stimuli. This decaying rate is a consequence of the combined effect of: (i)~lost or shifted interest, and (ii)~competition. Note that we do not model the detailed interactions between influencers, but we represent the combined effect of all competitors on the same topic. 

The second term represents the change in popularity due to the emission of a post. Here~$N^{[i]}_I(t)$ is the counting process induced by the post-emission.
To ensure generality, and as supported by empirical evidence detailed in Section~\ref{subsec:charas}, $N^{[i]}_I$~is modeled as a point process admitting a conditional intensity (also referred to as \textit{stochastic intensity}) 
which may depend on the effective popularity, i.e.,~$\lambda_i(t):= \lambda_i \left(X^{[i]}(t)\right)$.\footnote{$N^{[i]}_I$ can be made a homogeneous Poisson process by setting $\lambda_i(t):= \lambda_i$.}
A post published at time~$t$ typically attracts audience attention, leading to an increase of influencer~$i$'s popularity (hereafter, we will informally refer to such an increase in popularity as \textit{jump}), modeled by random variable~$V_i(t)$.
Note that the distribution of~$V_i(t)$ depends on factors such as post quality and attractiveness, as well as the size and engagement of interacting users. 
Actually the reach of a post is determined through multiple mechanisms: platform-specific algorithmic curation (such as EdgeRank) designed to optimize engagement metrics, subsequent redistribution via user sharing behaviors, and direct exposure to the influencer's established follower base.
Occasionally, posts \textit{go viral}, achieving unexpected success beyond followers. For simplicity and generality of the model, we abstract away the specific mechanisms by which posts attract attention.
Instead, we limit ourselves to statistically characterizing the cumulative effect of such mechanisms, making~$V_i(t)$ simply dependent on the current popularity of influencer~$i$ at time~$t$,~$X^{[i]}(t^-)$.

Finally (third term), external factors can affect influencer popularity. 
Exogenous events may cause popularity jumps, denoted by~$W_i(t)$. Due to their external nature, $W_i(t)$~are reasonably independent of influencer~$i$'s behavior on the platform. The arrival process~$N^{[i]}_E$ can be modeled as a Poisson process. 

\subsection{Statistical Characterization and Main Assumptions}\label{subsec:charas}
We now characterize the key random variables in Eq.~\eqref{SDE} through empirical validation on Facebook data. Then, we formalize our findings as a set of assumptions.

\subsubsection{Post Popularity Jumps.} 
First, we define the success of a post at time $t$ as the number of likes (positive reactions) it receives.
This number corresponds to the popularity jump~$V_i$ at time~$t$. 
Then, we show its dependence on the instantaneous popularity~$X^{[i]}(t^{-})$. We reconstruct the popularity dynamics~$X^{[i]}$ based on Eq.~$\eqref{SDE}$, beginning with a zero initial condition. 
We identify the empirical posting time sequences with the sample paths of the point process~$N^{[i]}_I(dt)$, and we set~$N_E^{[i]}(dt)=0$ due to the lack of information about exogenous events.
We then derive a normalized popularity~$\tilde{X}^{[i]}(t)$ by dividing the popularity by its maximum value. 
This normalization is necessary to aggregate data from various influencers whose absolute popularity levels can differ significantly. Similarly, we use a normalized version of the empirical success~$\tilde{V_i}$ of influencer~$i$'s posts.
Figure~\ref{fig:assumption_cond_expect}(a) shows the average normalized jump conditioned on the instantaneous normalized popularity. We partitioned the range of normalized popularity into~$10$ bins and calculated the mean number of likes for posts falling within each bin. The results, presented for different values of parameter~$\gamma$ appearing in equation \eqref{SDE} ($\gamma \in \{32, 128, 512\}$ days), are demonstrating a steady upward trend with increasing normalized popularity.
As~$\gamma$ decreases (i.e., $1/\gamma$ increases), the dependency becomes smoother, likely due to the increased inertia in influencer 
popularity dynamics.

\begin{figure}[ht]
    \centering
    \begin{tabular}{cc}
        \includegraphics[width=0.41\textwidth, valign=t]{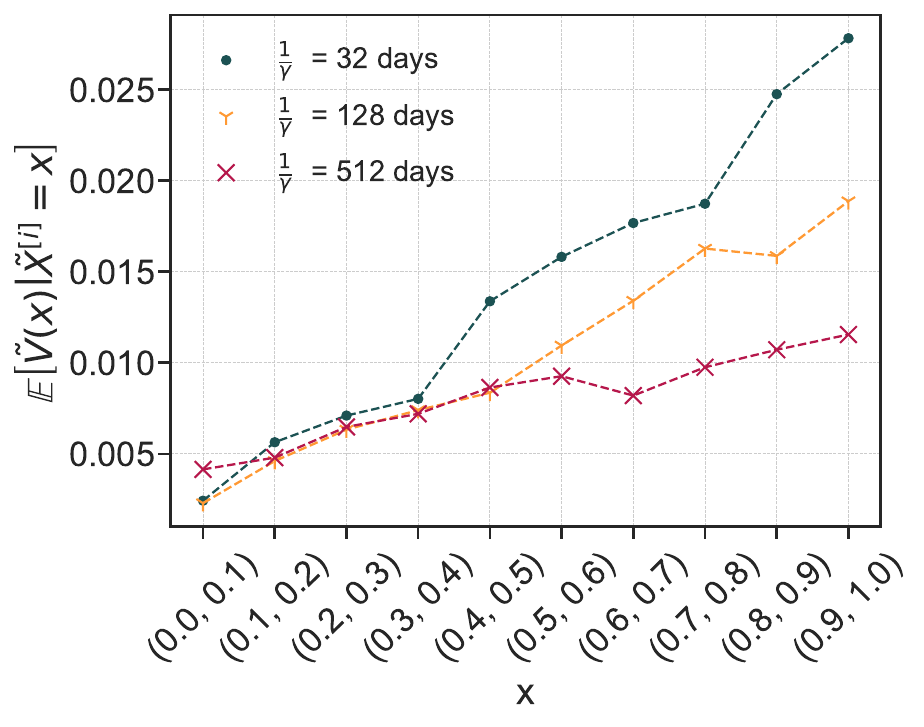}
        \hspace{0.5cm}&\hspace{0.4cm} 
        \includegraphics[width=0.46\textwidth, valign=t]{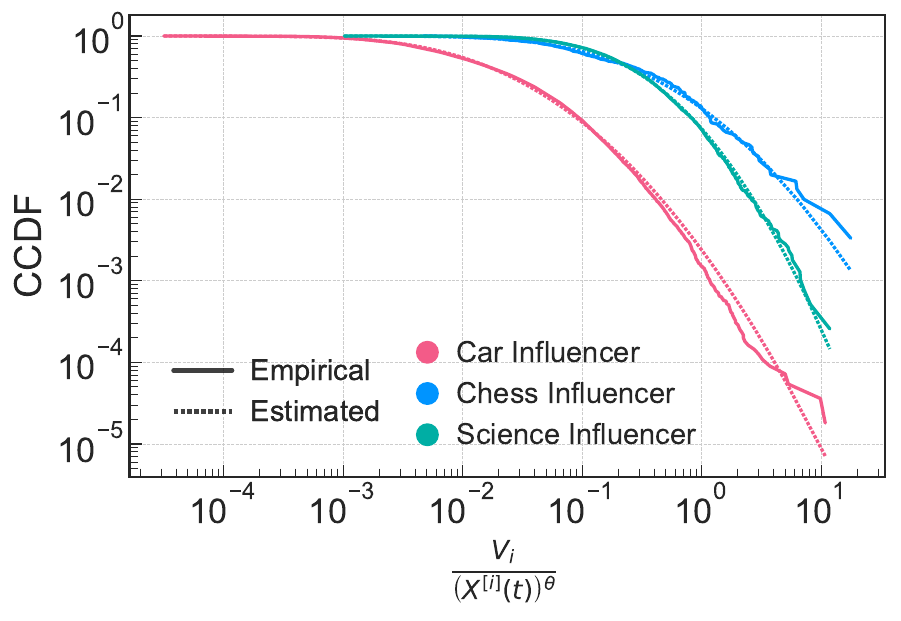} \\
       { \scriptsize $\quad\quad$ (a)}  & { \scriptsize $\quad\quad\quad$ (b) } \\
    \end{tabular}
    \caption{(a) Empirical conditional expectation of the normalized jump as a function of the normalized effective popularity. (b) Empirical and estimated CCDF of the ratio between the \textit{number of likes} and~$(X^{[i]})^{\theta}$ for three Facebook influencers.
   }
    \label{fig:assumption_cond_expect}
\end{figure}

Based on these empirical observations, we assume that~$\mathbb{E}[V_i(t)\mid X^{[i]}(t^-)=x]$ is an increasing function of $x$. In particular, we make the following hypothesis:
\begin{assumption}\label{ass:cond_expectation}
$\mathbb{E}[V_i(t)\mid X^{[i]}(t^-)=x]=\varepsilon + \beta_i x^\theta$,
where~$\epsilon$ is an arbitrarily small positive constant to avoid absorbing states and $\beta_i>0$. 
\end{assumption}
{Parameter~$\beta_i$ is influencer-specific and reflects the influencer's competence on the topic and their ability to create captivating content. 
In contrast, 
we assume $\theta$ to be primarily determined by users' behavior and platform engagement algorithms, and thus influencer independent.}
                                   
It is worth mentioning that the publication of a controversial post might lead to a decrease in popularity (negative jump). However, since these events are statistically rare and hardly detectable in our dataset, we decided to neglect them.

\smallskip

For the conditional distribution~$F_{V_i}(z\mid x)$, we make the simplest possible assumption by considering that~$V_i$ scales proportionally to~$\varepsilon + \beta_i (X^{[i]}(t^-))^\theta$: 

\begin{assumption}\label{ass:salti}
Let~$\hat{V}_i$ be a positive random variable independent from~$X^{[i]}(t^-)$, with unitary mean.
Then $V_i=\left(\varepsilon + \beta_i (X^{[i]}(t^-))^\theta\right) \hat{V}_i$. 
\end{assumption}
It should be noticed that, according to Assumption \ref{ass:salti}, the conditional partition function is given by~$F_{V_i}(z\mid x)= F_{\hat{V}_i}(\frac{z}{\varepsilon+\beta_i x^\theta})$.

\smallskip

Our goal is to determine the distribution of success from empirical data, again considering the number of received likes as a proxy of the post's success. To this end, we have selected a set of \textit{candidate} distributions to be tested, including: exponential, lognormal, and power-law. 
A few parameters need to be first identified. We must differentiate between success-specific parameters (e.g.,~$\beta_i$), which reflect influencer posting behavior, and system-level parameters (i.e.,~$(\gamma$, $\theta)$), which capture broader user and platform dynamics.
While success-specific parameters can be estimated from the data (for each chosen distribution), system-level parameters cannot be directly extracted from traces. Therefore, we have developed a Maximum Likelihood Estimation~(MLE) procedure that, given both a candidate distribution and the pair~$(\gamma, \theta)$, first finds the best success-specific parameters.   Once the parameters have been obtained, we evaluate the quality of the fitting by computing the Kolmogorov distance~$\kappa$ between the synthetic and the empirical distribution. Recall that the Kolmogorov distance between two distributions~$F$ and~$G$ is defined as~$\kappa(F, G) = \sup_{x \in \mathbb{R}} |F(x) - G(x)|$.
Since the lognormal distribution results the best one (i.e., the one minimizing $\kappa$) for over~98\% of the influencers under several choices of pair~$(\gamma, \theta)$ (as better detailed below), we 
decided to consistently adopt the lognormal shape for  
distribution~$\hat{V}_i$.

Recall that a lognormal distribution has two parameters. 
For simplicity, we neglected~$\varepsilon$ and focused on~$\beta_i \hat{V}_i$, which is, as consequence of Assumption~\ref{ass:salti}, a lognormal random variable with mean $\beta_i$ (recall $\mathbb{E}[\hat{V}] = 1$) and coefficient of variation $CV_i$. The above MLE procedure can be used to estimate both parameters for each influencer.
In Figure~\ref{fig:assumption_cond_expect}~(b), we report the result of the fitting for three influencers, one for each of the considered domains, comparing the empirical Complementary Cumulative Distribution Function~(CCDF) with the sythetic lognormal distribution for a particular choice of~$(\gamma,\theta)$, as discussed below.

Finally, we argued that system-level parameters ($\gamma$ and $\theta$) cannot be directly measured. However, we can determine them indirectly. First note that, given a choice of~$(\gamma, \theta)$ and performing the~MLE, it is possible to compute the Kolmogorov distance~$\kappa_i (\gamma, \theta)$ between the synthetic and empirical distribution. This provides a measure of the goodness of the fit. 
At this point, we select the best system-level parameters as those that minimize the cumulative Kolmogorov distance (over all influencers),
namely: $(\gamma^*, \theta^*) = \min_{\gamma, \theta}\sum_i \kappa_i (\gamma, \theta)$.

Figure~\ref{fig:new-heatmap}  reports values of $\sum_i \kappa_i (\gamma, \theta)$  for different choices of $(\gamma, \theta)$ . Observe that the best fitting of the empirical data is obtained for  $\theta^*=0.7$ and $\gamma^*=1/128$~[1/days]. These are also the parameters used in Figure~\ref{fig:assumption_cond_expect}~(b).

In Table \ref{tab:elo_beta}, we present our estimates of $\beta_i$  
for five prominent influencers within the chess domain, alongside their Elo 
ratings (for standard chess), widely recognized as the definitive measure of chess playing strength.
Notably, the correlation between $\beta_i$ and Elo is not particularly 
strong. This suggests that technical expertise in the topic is just one of several factors contributing to an influencer's success on social media platforms.

\begin{figure}[htbp]
    \centering
    \begin{minipage}{0.58\linewidth}
        \centering
        \includegraphics[width=0.8\textwidth]{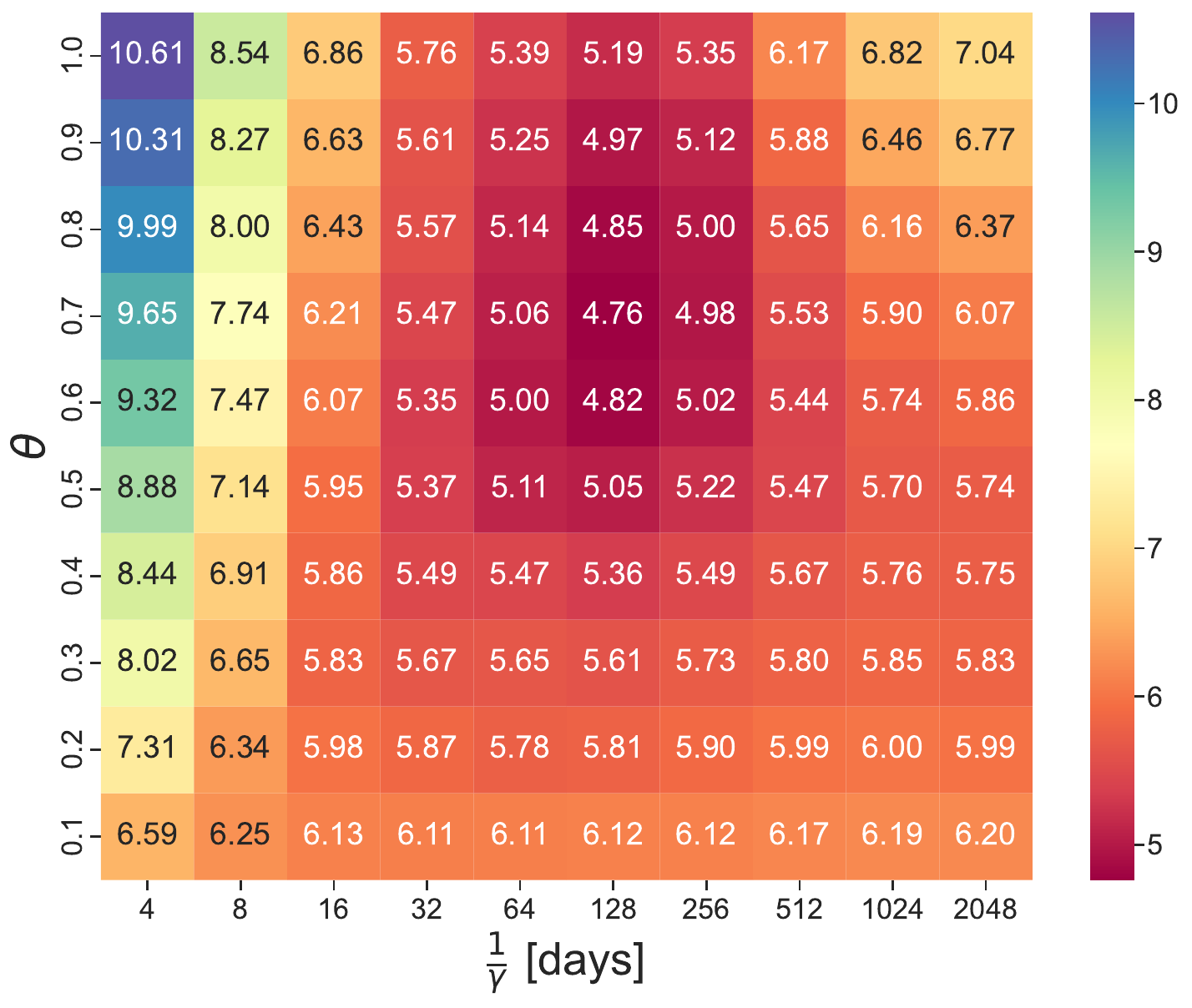}
        \caption{Cumulative Kolmogorov distance between lognormal fit and empirical distribution.}
        \label{fig:new-heatmap}
    \end{minipage}
    \hfill
    \begin{minipage}{0.38\linewidth}
        \centering
        \captionof{table}{Comparison between the influencer's estimated~$\beta_i$ and their FIDE Elo score.}
        \label{tab:elo_beta}
            \scalebox{0.9}{ 
                \begin{tabular}{l|c|c}
                    \toprule
                    \textbf{Influencer} & \textbf{Elo} & \textbf{$\beta_i$} \\
                    \midrule
                    Judit Polgar & 2675 & 0.424 \\
                    Alexandra Botez & 2044 & 0.207 \\
                    Magnus Carlsen & 2837 & 0.871 \\
                    Fabiano Caruana & 2776 & 1.194 \\
                    Tania Sachdev & 2396 & 1.953 \\
                    \bottomrule
                \end{tabular}
            }

    \end{minipage}
\end{figure}

\subsubsection{Intensity of the Posting Process.}

The posting patterns of influencers represent another critical factor shaping popularity dynamics.
Many studies assume the posting process to be a homogeneous Poisson process, and indeed this simple process has been found to describe well most of the activity~\cite{Crane2008}.
However, data traces exhibit pronounced bursty patterns, indicating that the Poisson assumption does not always align with real-world data.

Recall that in a Poisson process, the number of events occurring within a fixed time window is distributed as a Poisson random variable. The index of dispersion of a Poisson distribution (i.e., the ratio between the variance and the mean of the distribution) is equal to~1.
We computed the index of dispersion~$D$ of the number of posts in one-week windows, for each influencer. In Figure~\ref{fig:dispersion_rate}~(a), we report the distribution of~$D$, whose values are typically much larger than~1, suggesting that the posting process is more complex than Poisson. 

Let $\{\tau_n\}_{n}$ be the sequence of inter-post time, i.e., the time interval between two consecutive posts by the same influencer.
Figure~\ref{fig:dispersion_rate}~(b) shows the empirical estimate of the reciprocal of the average next inter-post time, conditionally over the normalized (w.r.t. the maximum) instantaneous popularity~$\tilde{X}^{[i]}(t)$ of all influencers.
This estimation represents an approximation of the stochastic intensity~$\lambda(t) | X^{[i]}(t)=x$. Even if the trend is noisier than that in Figure~\ref{fig:assumption_cond_expect}, there is a clear positive correlation between the posting rate and the effective popularity. The correlation becomes weaker as the value of~$\gamma$ decreases ($1/\gamma$ increases), due to the larger system inertia. 
Motivated by these observations, we make the following hypothesis.

\begin{figure}[h!]
    \centering
    \begin{tabular}{cc}
        \includegraphics[width=0.435\textwidth, valign=t]{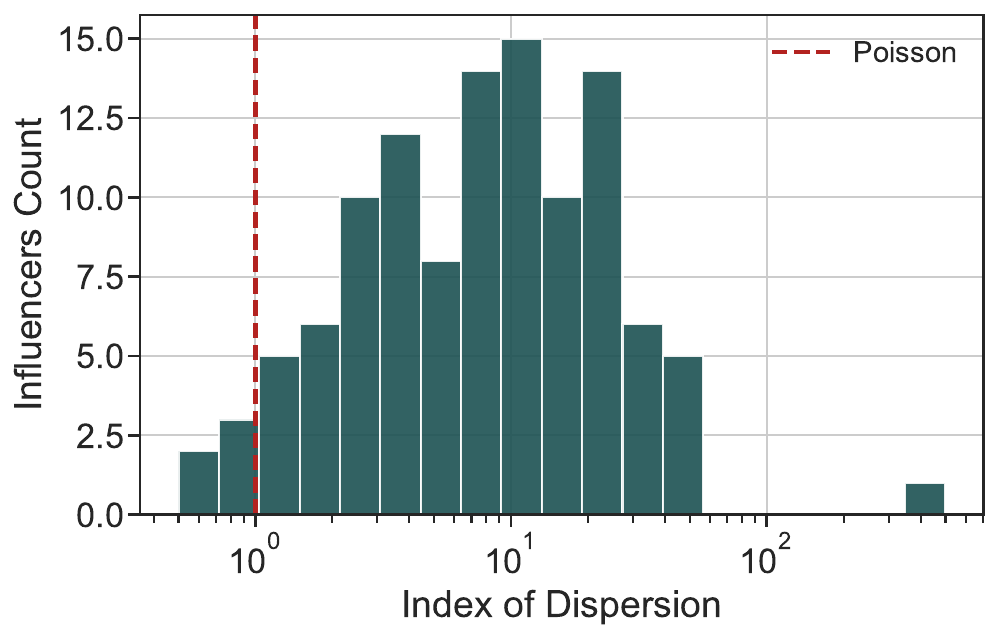} 
        \hspace{0.55cm}&\hspace{0.55cm}
        \includegraphics[width=0.41\textwidth, valign=t]{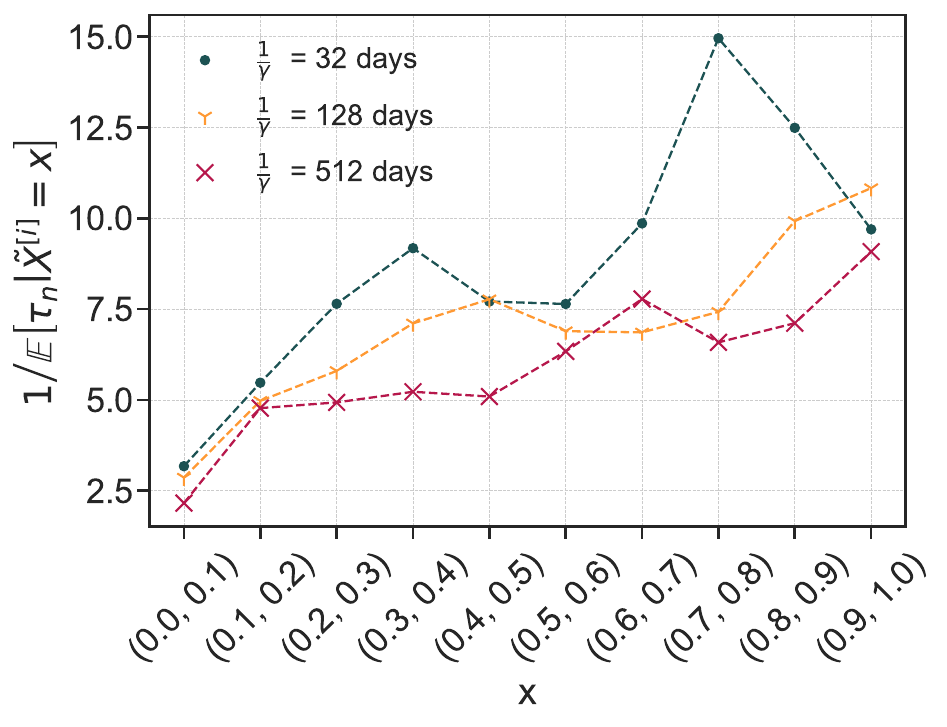} \\
        {\scriptsize $\quad\quad$ (a) } & {\scriptsize $\quad\quad\quad$ (b) } \\
    \end{tabular}
    \caption{(a) Distribution of the influencer's index of dispersion of the number of arrivals in a one-week window over the entire time horizon. (b) Instantaneous stochastic rate for different choices of~$\gamma$, measured in posts/day.}
    \label{fig:dispersion_rate}
\end{figure}

\begin{assumption}\label{ass:posting_rate} For any influencer $i\in\mathcal{I}$, the conditional stochastic intensity of $N^{[i]}_I(dt)$ follows the law: $\lambda^{[i]}(t)\mid \{X^{[i]}(t)=x\}= \lambda_{0,i} + \lambda_{1,i} x^{\phi_i}$.
\end{assumption}

This assumption maintains reasonable flexibility while capturing the observed pattern that influencers' posting rate tends to increase with growing popularity. However, the specific parameters characterizing this relationship vary among influencers, reflecting individual behavioral factors.

\subsection{Discussion on Parameters and their Relationships}
We summarize all of the model parameters in Table~\ref{tab:sce}.
Note that $\beta_i$ can be regarded as the average strength of the content generated by influencer $i$, and it subsumes two different aspects: (i) the intrinsic average quality of their posts; (ii) the level of engagement of their audience, which depends on behavioral features of  $i$, as well as, from intrinsic characteristic of their audience. 

Parameter $\theta$, instead, accounts for the effect of the platform's engagement mechanism(s)  to determine the audience of individual posts (and for such a reason it is assumed independent of $i$).
Parameters $\lambda_{0,i}$, $\lambda_{1,i}$ and   $\phi_i$  are tightly related to behavioral traits of influencer $i$ (like impulsivity).  At last, parameter $CV_i$ quantifies the variability in post success.

It is important to note that these parameters exhibit interdependencies. For example, we can anticipate an inverse relationship between $\beta_i$
and the posting rate $\lambda_i$ (determined by $\lambda_{0,i}$, $\lambda_{1,i}$ and $\phi_i$):  if an influencer posts too frequently, 
this typically compromises content quality due to constrained resources and time, while potentially fatiguing their audience, leading to a reduction of post effectiveness (represented by $\beta_i$). The above natural feedback mechanism inherently discourages influencers from adopting overly aggressive posting strategies.

Furthermore, a nuanced relationship exists between the popularity decay rate~$\gamma$ and content publication parameters. 
When $\gamma$ is large -- whether due to natural decline of audience attention or competition from alternative content -- influencers are expected to increase their posting frequency (within the previously discussed constraints) 
to prevent the loss of influence and visibility.
Conversely, a smaller popularity decay
rate reduces the urgency to post frequently, as popularity remains more stable over time, weakening the correlation between posting rate and instantaneous popularity.

\begin{table}[h!]
    \centering
    \caption{Parameters summary and reference scenario for the numerical analysis.}
    \label{tab:sce}
   \resizebox{1\columnwidth}{!}{ \begin{tabular}{cp{9cm}c}
        \toprule
        \textbf{Symbol} & \textbf{Description} & \textbf{Reference value} \\
        \midrule
        $\beta_i$ & ability of influencer $i$, $i \in \mathcal{I}$ & $0.9^{i-1}$ \\
        $\gamma$ & popularity discount rate & 1/64 (days) \\
        $\theta$ & exponent for non-linear dependence of jumps on popularity & 0.6 \\
        $\epsilon$ & small constant to avoid absorption in zero & 0.01 \\
        $\lambda_{0,i}$ & constant term of posting rate & 4 (posts/day) \\
        $\lambda_{1,i}$ & scale factor of the variable term of posting rate & 0 \\
        $\phi_i$ & exponent of the variable term of posting rate & 0 \\
        $CV_i $ & coefficient of variation of the distribution $\hat V_i$ & 4 \\
        \bottomrule
    \end{tabular}}
\end{table}

\subsection{Theoretical Results}\label{sec:theory}
In this section, we carry out a complete probabilistic analysis of the effective popularity~$X(t)$ as specified in Eq.~\eqref{SDE}. 

First, observe that Eq.~\eqref{SDE} defines a homogeneous Continuous Time Markov Process over a general space state.  We denote with $\{X_n\}_{n\in \mathbb{N}}$ the embedded Discrete Time Markov Process  (DTMP), obtained by sampling $\{X(t)\}_{t\in \mathbb{R}_+}$ at jump times, i.e., $X_n=X(T_n^-)$, where  $\{T_n\}_{n\in \mathbb{N}}$ are the ordered points of $N_I+N_E$.   We refer the reader to \cite{Twedie-Meyn} for a comprehensive analysis of Markov processes over general space state.
Denoted with  $\mathcal{B}(\mathbb{R}_+)$ the family  of Borellian sets  over $\mathbb{R}_+$, 
and with $P^n(x,A)=\mathbb{P}(X_n\in A\mid X_0=x)$ for $A\in  \mathcal{B}(\mathbb{R}_+)$, $n \in \mathbb{N}\setminus \{0\}$,  the $n$-step  Markov kernel of the DTMP, we have the following result. 
	\begin{theorem}\label{lemma-ergo}
  Let Assumption \ref{ass:salti}-\ref{ass:posting_rate} be satisfied with $\epsilon>0 $, $\lambda_0>0$. If $\theta+\phi<1$, or $\theta+\phi=1$ and $\beta(\lambda_0+\lambda_1)$ is sufficiently small with respect to $\gamma$, then
		$\{X_n\}_n$  admits
    a unique stationary probability measure, denoted by $\pi$.
     Moreover, for any  initial condition, as $n$  grows large 
     \mbox{$\sup_{A\in \mathcal{B}(0,\infty)}| P^n(x, A)- \pi(A)|\to 0$}.
	\end{theorem}	
The detailed proof is reported in the Appendix.
In short, we show  $\mu_{\text{Leb}}$-irreducibility\footnote{$\mu_{\text{Leb}}$ denotes the Lebesgue measure}, as well as strong aperiodicity for the embedded  DTMC $\{X_n\}_n$, using a rather direct approach.
Then, Harris recurrence can be proved using standard drift arguments, i.e., by adopting $\mathcal{L}(X_n)=X_n$ as Lyapunov function \cite{Twedie-Meyn}.

Theorem~\ref{lemma-ergo} guarantees that, under the appropriate conditions on system parameters, the Markov process admits a unique stationary probability measure; the process is ergodic, implying convergence to the stationary distribution regardless of the initial condition.
This means the Markov chain is ergodic, and time averages of observables will converge to the expected values under the stationary distribution.

At last, with rather standard arguments, it can be shown that the ergodicity of~$\{X_n\} _n$ implies  the ergodicity of~$\{X(t)\}$
(i.e. a unique stationary distribution~$\Pi(A)$  exists and~$P^t(x,A)\to \Pi(A)$  for every~$x\in \mathbb{R}^+$  and~$A\in \mathcal{B}(\mathbb{R}^+)$.
		   
Denoting by~$F(y,t):=\mathbb{P}(X(t)\leq y)$ the partition function 
of~$X(t)$,~$f(x, t)$ the associated probability density and~$\mu$ the rate of process~$N^{[i]}_E$, we have:
\begin{theorem}
\label{thm:Kolomgorov} Under the assumptions of Theorem \ref{lemma-ergo},
Function ${F}(y,t)$ satisfies the following Partial Integro-Differential Equation
\begin{align}\label{diff-eq}
\frac{\partial F(y,t)}{\partial t} =  \gamma y \frac{\partial F(y,t)}{\partial y}-  \int   [\lambda(t)\bar F_V( y-x \mid x) +
\mu  \bar F_W( y-x)]\frac{\partial F(x,t) }{\partial x}     \mathrm{d}x
\end{align}
where  $\bar F_V(\cdot\mid x)$ and $\bar F_W(\cdot)$ are respectively the conditional CCDF of $V$ and the CCDF of $W$, respectively.
Moreover, there exists a unique stationary distribution $F(y)$ satisfying the following ordinary integro-differential equation:
\begin{equation}\label{stat-diff-eq}
    \gamma y \frac{\mathrm{d} F(y)}{\mathrm{d} y}=\int  [ \lambda(t)\bar F_V( y-x \mid x)+ \mu\bar  F_W( y-x ) ] \frac{\mathrm{d}  F(x) }{\mathrm{d}  x}     \mathrm{d}x.
\end{equation}
\end{theorem}
Theorem~\ref{thm:Kolomgorov}  allows the evolution of the distribution function~$F(y,t)$ over time to be traced through a partial integro-differential equation.
In addition, the result also provides the stationary distribution, described through an Ordinary Integro-Differential Equation, which describes the system at equilibrium.
Again, the proof is omitted due to space constraints and will appear in a forthcoming version of the paper.
To corroborate our results, we will present in Section~\ref{sec:num_results} a comparison between the simulation of the Markov process and the theoretical distributions obtained through Theorem~\ref{thm:Kolomgorov}.

\section{Numerical Analysis}\label{sec:num_results}

In this section, we use our model to explore how key performance indicators depend on various system parameters. We focus on a reference scenario derived from our empirical data, and perform a sensitivity analysis by varying one parameter at a time. This allows us to investigate what-if scenarios and better understand popularity dynamics and influencer competition.

\subsection{Metrics and Sensitivity Analysis}
We consider a set of five virtual influencers competing for user attention, $\mathcal{I} = \{1,2,3,4,5\}$. 
We set $\beta_i = 0.9^{i-1}$ to introduce a systematic variation in their ability to garner user engagement.
Note that, although the first influencer is expected to be the most popular one on average, the model allows for the emergence of any influencer as the dominant actor, at least temporarily.

To quantitatively assess the competitive dynamics and performance hierarchy among these influencers, we introduce the following two key metrics:

\begin{itemize}
\item {\bf First place probability}, defined as $\pi_1^{[i]} = \mathbb{P}(X^{[i]} \geq X^{[j]}, \forall j \neq i)$, which is the probability that influencer $i$ achieves the highest popularity under stationary conditions. Conceptually, it can also be interpreted as the long-term fraction of time during which the given influencer occupies the top position in the popularity ranking.

\item {\bf First place average stay}, denoted with $S_1^{[i]}$, which is the average sojourn time of influencer $i$ in the first place (in stationary conditions),
providing insight into the temporal stability of its dominant status.
\end{itemize}

Inspired by our empirical analysis of data traces, we consider a reference scenario whose parameters are summarized in Table \ref{tab:sce}. 
We remark that in our reference scenario we have made a few simplifying assumptions to obtain a simple baseline case:
(i) the posting process is the same for all influencers (parameters
$\lambda_0, \lambda_1, \phi$ are now independent of $i$);
(ii) the posting process is a simple Poisson process ($\lambda_1 = 0$);
(iii) while the mean number of likes obtained by a post is different for each influencer, the coefficient of variation $CV$
of the lognormal distribution is here assumed to be the same for all influencers.

Figure~\ref{fig:pdfx} shows, on a log-log scale, the stationary Probability Distribution Function (PDF) of $X^{[i]}$ for the five considered influencers, obtained by simulation and by numerically solving~\equaref{stat-diff-eq}. The perfect match between simulation and analysis cross-validates both approaches to obtain the
stationary distribution of influencers' popularity.

Starting from the baseline case, Figure \ref{fig:gam} explores the impact of  popularity decay rate $\gamma$ on $\pi_1^{[i]}$ (left plot) and
$S_1^{[i]}$ (right plot). As expected, as $\gamma$ diminishes ($1/\gamma$ increases), the top influencer tends to monopolize users' attention, since its popularity decays slower, reflecting the cumulative effect of a large number of posts, which tends to concentrate around the mean (determined by intrinsic ability $\beta_i$). The opposite is true for large $\gamma$ (small $1/\gamma$), where we observe the opposite regime in which all influencers demonstrate comparable probabilities of achieving dominance, as popularity dynamics in this case are principally governed by the stochastic success of individual posts.

\begin{figure}[!ht]
    \centering
    \includegraphics[width=0.45\linewidth]{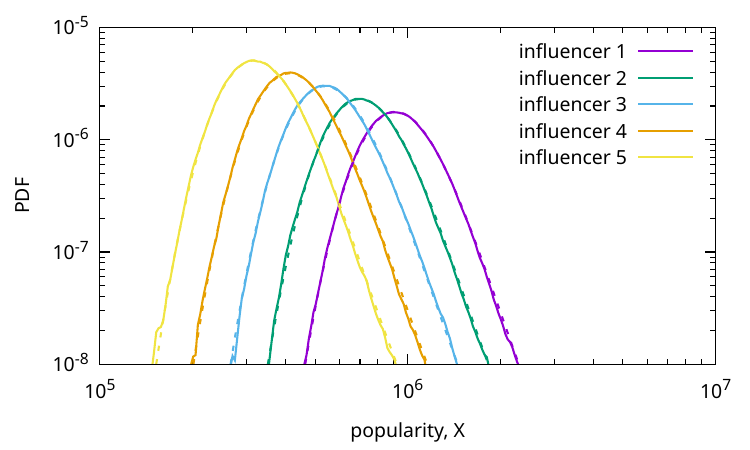}
    \caption{Popularity distribution of five influencers in the baseline scenario. Comparison between simulation results (solid) and analytical results (dashed lines).}
    \label{fig:pdfx}
\end{figure}

\begin{figure}[!ht]
    \centering
    \includegraphics[width=0.45\linewidth]{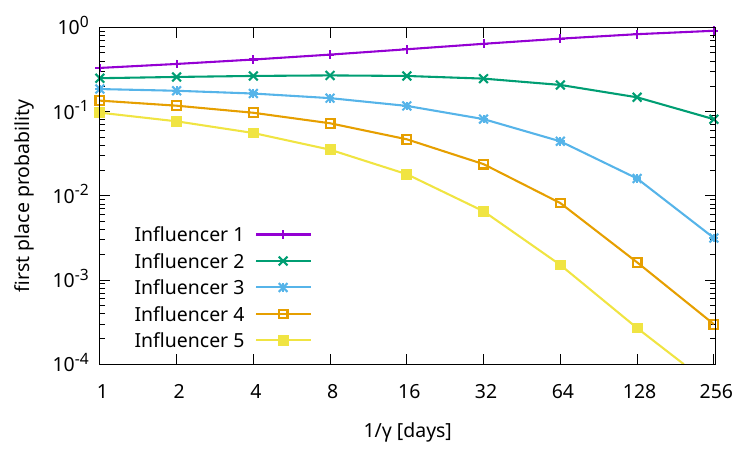}
    \hspace{8mm}
    \includegraphics[width=0.45\linewidth]{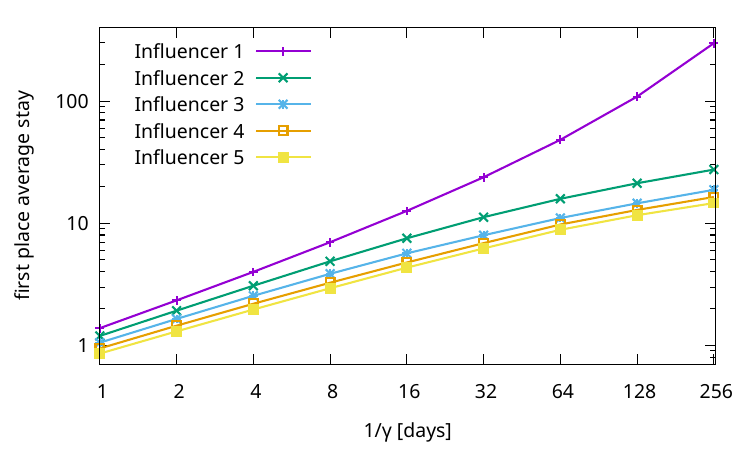}
    \caption{Sensitivity analysis with respect to popularity decay rate $\gamma$.
    Impact on first-place probability (left plot) and first-place average stay (right plot). }
    \label{fig:gam}
\end{figure}

\begin{figure} [!ht]
    \centering
    \includegraphics[width=0.45\linewidth]{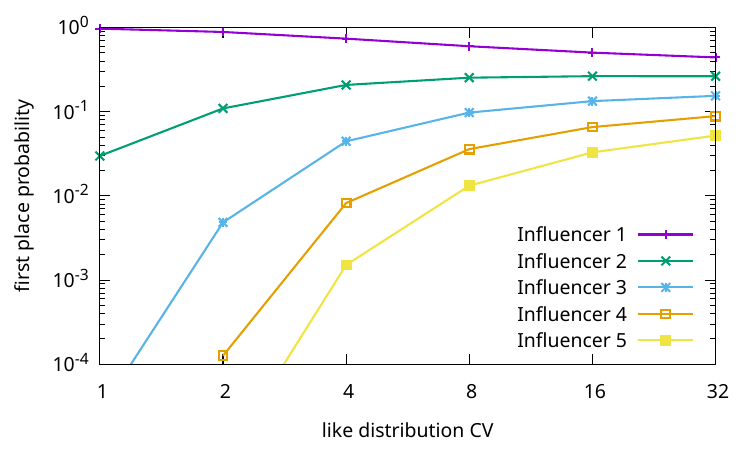}
    \hspace{8mm}
    \includegraphics[width=0.45\linewidth]{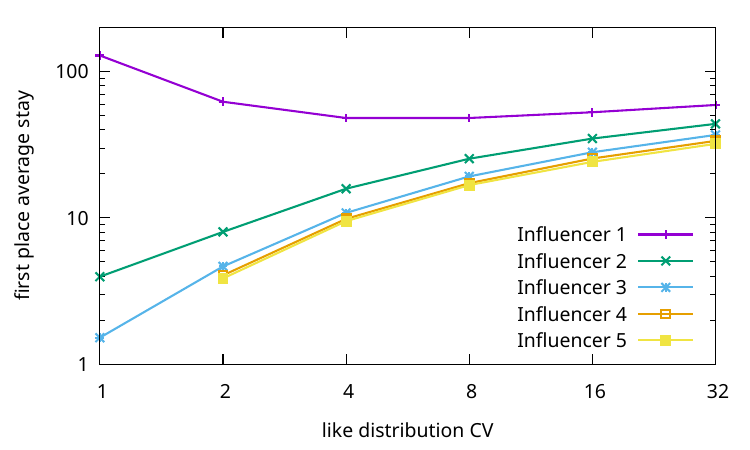}
    \caption{Sensitivity analysis with respect to the coefficient of variation $CV$ of the like distribution. Impact on first place probability (left plot) and first place average stay (right plot). }
    \label{fig:CV}
\end{figure}

Analogous considerations can be done when we fix $\gamma$ and change the coefficient of variation of the like distribution (see Figure \ref{fig:CV}). 
Naturally, concurrent variation of both $\gamma$ and $CV$ relative to our baseline would produce a compound effect, to be again interpreted 
through the lens of stochastic versus deterministic behavior (i.e., concentration around the mean) in the resulting popularity dynamics.

We emphasize that a proper notion of fairness among influencers is required to assess whether a given system performs better than another, but a formal definition of this notion extends beyond the scope of the present investigation.

\begin{figure}[!ht]
    \centering
    \includegraphics[width=0.45\linewidth]{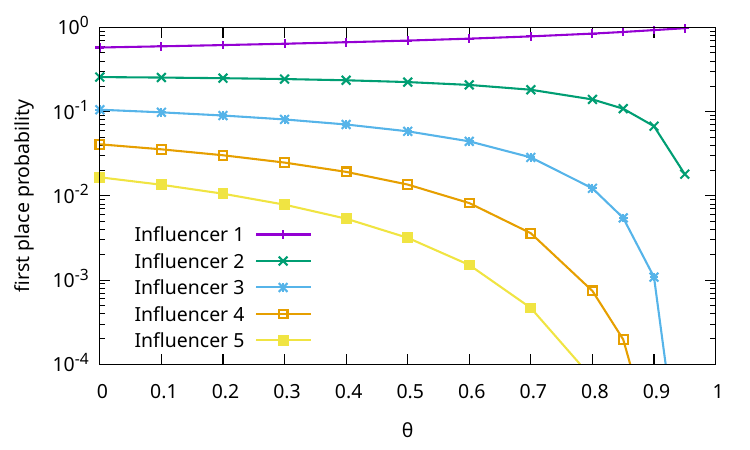}
    \hspace{8mm}
    \includegraphics[width=0.45\linewidth]{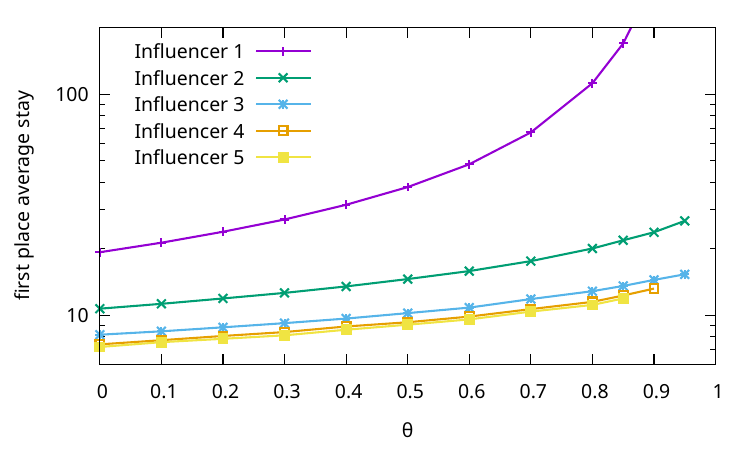}
    \caption{Sensitivity analysis with respect to exponent $\theta$. Impact on first place probability (left plot) and first place average stay (right plot). }
    \label{fig:theta}
\end{figure}

Next, in Figure \ref{fig:theta}, we consider the impact of exponent $\theta$, describing the (sub-linear) growth of the number of likes collected by a post as a function of the influencer's current popularity. 
As expected, larger values of $\theta$ (but recall that we need $\theta < 1$ for the system to be stable) amplify the disparity among influencers, underscoring the critical role of the platform's engagement mechanism(s) in determining their relative success.

\subsection{Impact of Non-Poisson Post Arrival Process}
Here, we examine how the burstiness of the post arrival process affects our metrics with respect to the baseline scenario, wherein posts are generated according to a simple Poisson arrival process of fixed intensity 4 posts/day. Recall that in our model we consider an inhomogeneous Poisson process for each influencer $i$, with intensity $\lambda(t\mid X^{[i]}(t^-)=x) = \lambda_{0,i} + \lambda_{1,i} x^{\phi_i}$, 
which increases the probability of new post emissions during periods of higher influencer popularity, a pattern consistently observed in empirical data.

To simplify the exploration of the parameter space, we proceed as follows:
(i) we establish the constant term of the posting rate $\lambda_{0} = 1$, equal for all influencers; 
(ii) for  $\phi$, we consider $\{0, 0.1, 0.2, 0.3\}$ as possible values, 
equal for all influencers; 
(iii) we systematically vary $\lambda_{1}$, setting it equal for all influencers.

Figure \ref{fig:lam1} shows the results of this experiment for $\phi = 0.2$, 
reporting the first place probability (solid lines, left y axes) and the obtained average posting rate (dashed lines, right y axes). As expected, larger values of $\lambda_1$ produce larger values of average posting rate, amplifying discrepancies among the five influencers. To isolate the effect of non-homogeneous post arrival rate,  while ensuring a fair competition among the influencers, we implemented the following methodology: for each considered value of $\phi$, 
and for each influencer, we numerically determined the value of $\lambda_{1}$ that yields an average posting rate of 4 posts/day (as in the baseline scenario). 

In Figure \ref{fig:lam1}, such values of $\lambda_{1}$, in the case of $\phi = 0.2$, are those at which dashed lines intersect the horizontal dotted line plotted at $y=4$. 
Subsequently, we recalculated the first-place probability in a unique scenario utilizing these calibrated $\lambda_{1}$, with results in Table \ref{tab:lam1}. 
 Note that the first column ($\phi = 0$) corresponds to the case of Poisson post arrival rate. We observe that, as we increase $\phi$, thereby enhancing the burstiness of the post arrival process, discrepancies among influencers tend to reduce.

\begin{figure}[htbp]
    \centering
    \begin{minipage}{0.46\linewidth}
        \centering
        \includegraphics[width=0.9\linewidth]{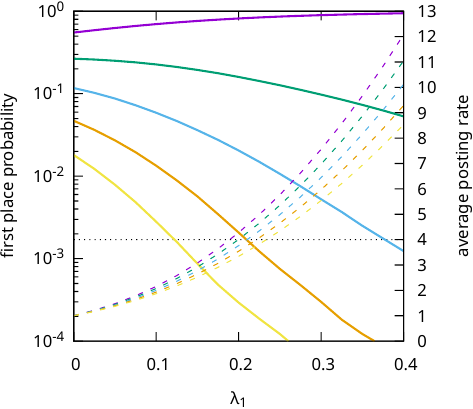}
        \caption{Sensitivity analysis with respect to $\lambda_1$ ,
        for fixed $\lambda_0 = 1$, $\phi = 0.2$.}
        \label{fig:lam1}
    \end{minipage}
    \hfill
    \begin{minipage}{0.50\linewidth}
        \centering
        \captionof{table}{First place probability with adapted $\lambda_{1}$, so that the average posting rate is constant (4 posts/day) for all influencers.}
        \label{tab:lam1}
        \begin{tabular}{c|c|c|c|c}
            \toprule
            Influencer & $\phi = 0$ & $\phi = 0.1$ & $\phi = 0.2$ & $\phi = 0.3$ \\
            \midrule
            1 & 0.738 & 0.701 & 0.651 & 0.597 \\
            2 & 0.208 & 0.226 & 0.253 & 0.264 \\
            3 & 0.045 & 0.058 & 0.073 & 0.099 \\
            4 & 0.008 & 0.012 & 0.019 & 0.032 \\
            5 & 0.002 & 0.002 & 0.004 & 0.009 \\
            \bottomrule
        \end{tabular}

    \end{minipage}
\end{figure}

\section{Discussion and Concluding Remarks}
Our work combines theoretical modeling of popularity dynamics with empirical data from Facebook, providing insights into the complex interaction of factors like individual activity, popularity decay, content attractiveness, and the platform's role in content visibility. Besides deriving ergodicity conditions, we conducted a sensitivity analysis to explore the impact of system parameters on possible outcomes. 

In principle, an ideal system should provide each influencer with a fair opportunity to gain attention, commensurate with their intrinsic merit.
In this context, one should aim to avoid two problematic extremes: (i) a situation in which the most attractive influencer monopolizes attention, leaving too little visibility for others; and (ii) a condition in which influencers are roughly equally likely to gain primacy, 
irrespective of their intrinsic merit or capabilities.
Fairness, therefore, involves not only fostering an environment of discussion and information exchange in which each influencer has a chance to gain public attention commensurate with merit, but also ensuring that the attention amplification induced by recommendation systems and platform algorithms does not distort competition.
Finding an optimal balance between the above two extremes involves complex ethical, technological, sociological, and political issues that are beyond the scope of this analysis. However, the proposed model, by offering a quantitative tool for examining how various factors influence the considered metrics, allows one to explore “what-if” scenarios, yielding valuable insights into popularity dynamics and competitive processes within online social networks.

Subsequent research will extend the analysis to other platforms, develop a framework for influencer competition, and explore control strategies to foster equity in competitive environments.

\section*{Acknowledgments}
    This research has been supported by the European Union -- Next Generation EU, Mission 4, Component 1, under the PRIN project TECHIE: ``A control and network-based approach for fostering the adoption of new technologies in the ecological transition'' Cod. 2022KPHA24 CUP Master: D53D23001320006, CUP: B53D23002760006.

\bibliographystyle{IEEEtran}
\bibliography{sample}

\appendix

\section{Proof of the Theorems in Section \ref{sec:theory}}

Note that $\{X(t)\}_{t\ge 0}$ is a Continuous Time Markov  Process (CTMP) over a general (uncountable) space state (i.e., $\mathbb{R}^+:= (0, \infty$)).
Given the times $T_n$ of $n$-th jump, we consider the Discrete Time Markov Priocess (DTMP), denoted by  $\{X_n\}_{n\in\mathbb{N}}$, which is defined on the same space state \cite{Twedie-Meyn}
by $
X_n=X(T_n^{-})=X(T_n)$ where $X(t^{-})$ stands for the left limit $\lim_{u\uparrow t} X(u)$.
We first characterize the distribution of inter jumps times.
\begin{proposition}\label{prop:T_n}Let $\Delta T_n=T_{n+1}-T_n$ be the inter-jump time and $X(t^+)=\lim_{s\downarrow t} X(s)$. We have 
\begin{align*}
    &\overline F_{\Delta T_n {|\color{black}X(T_n^+)}}(\zeta\mid z) :=\mathbb{P}(\Delta T_n> \zeta \mid {\color{black}X(T_n^+)}=z)= \\
    &\exp\left(-(\lambda_{0}+\mu){\color{black}{\zeta}}-\frac{\lambda_{1} }{\gamma\phi} z^{\phi} 
    (1-\mathrm{e}^{-\gamma \phi\zeta})\right)
\end{align*}
and the conditional probability distribution is given by
\begin{align*}
f_{\Delta T_n{|X(T_n^+)}}(\zeta\mid z) &= (\lambda_{0} +\mu + \lambda_{1}  z^{\phi}) \\
&\cdot \mathrm{e}^{-\gamma \phi\eta} )\mathrm{e}^{-(\lambda_{0}+\mu)\zeta-\frac{\lambda_{1} }{\gamma\tau} z^{\phi} 
\left(1-\mathrm{e}^{-\gamma \phi\zeta}\right)}.
\end{align*}

\end{proposition}

\begin{proof}
We compute the distribution by noting that the event 
 $\{\Delta T_n>\zeta\}$, is equivalent to the event that no jumps occur in the interval $(T_{n},T_{n}+\zeta]$. Therefore, as rather immediate consequence of the stochastic intensity definition,  given $X(T_n^+)=z$, we have that $\Delta T_n$ is  distributed as:
    \begin{align*}
    &\overline F_{\Delta T_n{|\color{black}X(T_n^+)}}(\zeta\mid z)= \mathbb{E}[\mathrm{e}^{-\int_0^\zeta \lambda(\eta) \mathrm{d\eta} }\mid z]\\
    &=\mathrm{e}^{- \int_0^\zeta \left(\lambda_0+\mu+ \lambda_1\cdot(z\mathrm{e}^{-\gamma \eta})^\tau\right) \mathrm{d} \eta}=\mathrm{e}^{-(\lambda_0+\mu)\zeta-\frac{\lambda_1 }{\gamma\tau} z^\tau 
    (1-\mathrm{e}^{-\gamma \tau\eta})}
\end{align*}

The conditional probability distribution is then obtained by taking the derivative of cumulative distribution function $f_{\Delta T_n{|\color{black}X(t_n^+)}}(\zeta\mid z)=\frac{\mathrm d}{\mathrm{d}\zeta}[1-\overline F_{\Delta T_n {|\color{black}X(t_n^+)}}(\zeta\mid z)]$.
\end{proof}

Now, turning our attention to the DTMP $\{ X_n\}_n$, we have:

\begin{proposition}
Let $\{X_n\}_{n\in\mathbb{N}}$ be the embedded chain associated to {{$\{X(t)\}_{t\geq 0}$}. Then the sequence $\{X_n\}_{n\in\mathbb{N}}$  evolves according to the following recursion
\begin{align}\label{eq:dtmc}
X_{n+1}=&[X_n +V(X_n)\ind_{
\{T_n\in N_I \}} +W\ind_{\{T_n\in N_E\}} ] \nonumber\\
&\cdot \mathrm{e}^{-{ \color{black}\gamma}\Delta T_n(X_n+V(X_n))} 
\end{align}
with  transition kernel defined by}
\begin{align}\label{eq:kernel}
    P(x, dy)
    = & \frac{1}{\gamma y} \int^\infty_{y} \Big[ p^{(n)}_I (x) f_{V|X}(z-x\mid x)\\  + & p^{(n)}_E (x)   f_{W}(z-x)
        \Big] f_{\Delta T_n}\left(\frac{1}{\gamma}\log \frac{z}{y} \mid z \right) dz
\end{align}
{{where $p^{(n)}_I $ and $p^{(n)}_E$ are the probabilities of the events $\{T_n \in N_I\}$ and $\{T_n \in N_E\}$, respectively, satisfying 
\[
\begin{split}
 p^{(n)}_E(x) =\frac{\mu}{\lambda_{0}+\mu+\lambda_{1}(x^{\phi})}\\
  p^{(n)}_I(x) =1- \frac{\mu}{\lambda_{0}+\mu+\lambda_{1}(x^{\phi})}.
\end{split}
\]}}
\end{proposition}
\begin{proof} Note that by construction: $\ind_{\{T_n\in N_I \}} +\ind_{\{T_n\in N_E\}}=1$, and
\[
\begin{split}
  p^{(n)}_E(x) :=\mathbb{E}[\ind_{\{T_n\in N_I \}}\mid X_n=x] =\frac{\mu}{\lambda_{0}+\lambda_{1}(x^{\phi})+\mu}
\end{split}
\]
from which
\begin{align*}
  p^{(n)}_I(x) & :=\mathbb{E}[\ind_{\{T_n\in N_E \}}\mid X_n=x] \\
  &=1-p^{(n)}_E(x)=1- \frac{\mu}{\lambda_{0}+\lambda_{1}(x^{\phi})+\mu}.
\end{align*}

Therefore the value of the process immediately after the \(n\)-th jump is:
\begin{align*}
 X(T_n^+) = &
 X(T_n) + V(X(T_n))\ind_{
    \{T_n\in N_I \}} +W\ind_{\{T_n\in N_E\}}\\ = & X_n + V(X_n )\ind_{
    \{T_n\in N_I \}} +W\ind_{\{T_n\in N_E\}} 
 \end{align*}
and the process evolves according to the decay dynamics:
   $$
   X(t)\ind_{\{ T_n< t \le T_{n+1}  \}} = X(T_n^+) \mathrm{e}^{-\gamma (t - T_n)}\ind_{\{ T_n< t \le T_{n+1}  \}}.
   $$
   Evaluating the process at \(T_{n+1} = T_n + \Delta T_n\), we get:
   $
   X(T_{n+1}) = X(T_n^+) \mathrm{e}^{-\gamma \Delta T_n}
   $
 and by replacing the post-jump value into the above expression:
   \[
   X_{n+1} =  \Big(X_n + V(X_n )\ind_{
    \{t_n\in N_I \}} +W\ind_{\{t_n\in N_E\}} \Big)
  \mathrm{e}^{-\gamma \Delta T_n(X_n + V(X_n))}.
   \]
   whose distribution is given in Proposition \ref{prop:T_n}.
   Recalling the  definition of  the transition kernel as: 
\[
P(x, dy):=  \mathbb{P}(X_{n+1}\in dy \mid X_n= x)
\]
and conditioning on the event $\{
 X_n + V(X_n )\ind_{
    \{t_n\in N_I \}} +W\ind_{\{t_n\in N_E\}} =z\}$, from \eqref{eq:dtmc}
 we obtain for any $z>y$ 
 $$
\left\{\Delta T_n(z) \ge  \frac{1}{\gamma} \log \frac{z}{y}\right\}=\{X_{n+1}\le y \}
$$
Then the probability that $X_{n+1}$ falls within a small interval around $y$ is given by 
\begin{align*}
P(x, dy)= & \int_{0}^{\infty} \mathbb{P}(X_{n+1} \in dy \mid X_n = x,\\ & V_n \ind_{
    \{t_n\in N_I \}} +W\ind_{\{t_n\in N_E\}}=z-x)\\ 
 &\times \Big[ p^{(n)}_I(x)  f_{V|X}(z-x\mid x) + p^{(n)}_E (x)   f_{W}(z-x)
\Big] dz.
\end{align*}

Substitute $\Delta T_n(z)$ and express the integral over $z$ from $y$ to $\infty$:

\begin{align*}
P(x, dy) = & \frac{1}{\gamma y} \int_{y}^{\infty}  \Big[ p^{(n)}_I (x) f_{V|X}(z-x\mid x)+ p^{(n)}_E (x)   f_{W}(z-x)
\Big] \\ & f_{\Delta T_n\mid z}\left(\frac{1}{\gamma} \log \frac{z}{y} \mid z \right) dz.
\end{align*}
\end{proof}

%

It should be noted that the $P(\cdot, dy)$ is  continuous (and hence lower semi-continuous) as long as $y\in (0, \infty)$ and the total variation distance 
$\delta(F_V(z \mid x_1), F_V(z \mid x_2))\to 0$  as $|x_1-x_2|\to 0$. 

In the following to simplify the notation we denote with $Z(X_n )= V(X_n )\ind_{
    \{t_n\in N_I \}} +W\ind_{\{t_n\in N_E\}} $ and $f_Z(y\mid x)= p^{(n)}_I (x) f_{V|X}(y\mid x) + p^{(n)}_E (x)   f_{W}(y)$

\begin{Lemma}
    The following facts are true.
    \begin{enumerate}
    \item[a)] The DTMP $\{X_n\}_{n\in\mathbb{N}}$ is irreducible with respect to the restriction to $[\kappa_I,\kappa_S]$ of the Lebesgue measure $\mu_{\text{Leb}(0,\infty)} $ for any $\kappa_I>0$ and  $\kappa_S>\kappa_I$;
     \item[b)] The DTMP $\{X_n\}_{n\in\mathbb{N}}$ is strongly aperiodic.
     \end{enumerate}
\end{Lemma}
\begin{proof}	In order to prove irriducibility we show that,  for any set  $A\in \mathcal{B}([\kappa_I,\kappa_S])$
    such that $\mu_{\text{Leb}}(A)>0$ and any $x\in (0,  \infty)$,  there exists $k\in\mathbb{N}$ such that: 
    \[
    \mathbb{P}(X_{n+k}\in A\mid X_n= x)>\delta_{x,A}>0.
    \]
    Now, for any $\varepsilon\in(0,1)$,
    let   $z_{\varepsilon} (x) := \sup \{z:   1- F_{V|X}(z\mid x)> \varepsilon
    \}$,
    note that as a consequence of model assumptions $z_{\varepsilon}(x)$   is non decreasing with respect to its argument.
    Then, uniformly  for every  $y\in (\kappa_I, \min(z_{\varepsilon}(\kappa_I), \kappa_S)]$
    \[
    P(x, dy)> \delta_{x,\varepsilon}> 0
    \]
    as it can be easily checked by inspection.  Therefore  for any 
    $A\in \mathcal{B}([\kappa_I, \min(z_{\varepsilon}(\kappa_I), \kappa_S)])$.
    \[
    \mathbb{P}(X_{n+1}\in A\mid X_n= x):= \int_A P(x,dy)\ge  \delta_{x,\varepsilon} \mu_{\text{leb}}(A)
    \]	  
    Now   if     $z_{\sup}(0)>\kappa_S$ we are done, otherwise,
    consider  the $k$ -steps transition kernel, recursively defined as 
    \[
    P^k(x, dy)= \int_z P^{k+1}(x,dz)P(z,dy) 
    \]
    It turns out that  for any $y\in [\kappa_I, \min(k z_{\varepsilon}(\kappa_I), \kappa_S)]$
    \[
    P^k(x, dy)>\delta'_{x,\varepsilon}
    \]
    and therefore similarly  as before we can conclude that 
    \[
    \mathbb{P}(X_{n+k}\in A\mid X_n= x):= \int_A P^k(x,dy)\ge  \delta'_{x,\varepsilon} \mu_{\text{leb}}(A)
    \]
    for any $A\in \mathcal{B}([\kappa_I, \min(  k z_{\varepsilon}(0), \kappa_S)].$  
    
    To complete the proof  it is enough to observe that 
    for a sufficiently large $k$  necessarily $k z_{\varepsilon}(\kappa_I)>\kappa_S$. Then the item a) is proved.

    Since it can easily shown that   $\inf_{x\in [\kappa_I,\min(  z_{\varepsilon}(\kappa_I), \kappa_S) ]} \delta_{x,\varepsilon}= \delta_\varepsilon >0$, we deduce
    as in  \cite{Twedie-Meyn} (see definition on page 114) that 
    $\{X_n\}_n$ is  the strongly a-periodic.	
    Indeed   the set $[\kappa_I,  \min(z_{\varepsilon}(\kappa_I),\kappa_S)]$  turns out to be  $\nu_1$-small with respect to $\nu_1=  \delta_{\varepsilon} \mu_{\text{Leb}}([\kappa_I,\min (z_{\varepsilon}(\kappa_I),\kappa_S ) ])$.
\end{proof}	
Therefore by Theorem  $4.0.1$ in \cite{Twedie-Meyn}  there exist a unique maximal  measure with respect to which 
$\{X_n\}_n$ is irreducible. Such a measure  necessarily dominates $\mu_{\text{Leb}} (0,\infty)$. 
Moreover 
given an arbitrary open set $O\in \mathcal{B}((0,\infty))$ the transition probability
\[
P(x, O)=\int_O P(x,dy)  
\]
is clearly lower-continuous with respect to $x\in(0, \infty)$, therefore 
$\{X_n\}_n$  is  weak Feller. Now by Theorem  6.0.1 (iii) in  \cite{Twedie-Meyn},
since the  support of $\mu_{\text{Leb}}$ has clearly  a non empty-interior,
$\{X_n\}_n$ is a  T-chain.
At last   by Theorem  6.0.1 (ii) in  \cite{Twedie-Meyn} we deduce that every compact set is petite
(observe that any subset of a petite set is trivially petite).

Now we turn our attention to recurrence properties of our DTMP.

\begin{Lemma}
    As long as $\theta+\phi<1$ or $\theta+\phi=1$  and  $c/\beta>1$ with $c:=\frac{\gamma }{\lambda_0+ \lambda_1+\gamma}\alpha^{-\phi} F_Z\left((\alpha-1)x\mid x\right)$ for an arbitrary $\alpha>1$
    $\{X_n\}_n$  
    is  $\mu_{\text{Leb}}$-irriducible, stongly aperiodic  positive Harris recurrent.
    Therefore a unique stationary probability measure, $\pi$  exists, and   i.e. for any  inital condition as $n$  grows large the $\sup_{A\in \mathcal{B}(0,\infty)}| P^n(x, A)- \pi(A)|\to 0$
    (i.e. the MC ergodic).
\end{Lemma}

\begin{proof} 
    By Theorem 11.3.4 in \cite{Twedie-Meyn},  to prove that $\{X_n\}_n$   is  a  positive Harris recurrent MC , it is enough to apply a drift argument.
    In particular, given petite set $C$, if 
    a real valued function $\mathcal{L}(\cdot)$ $(0,\infty) \to [0,\infty]$
    can be found such that:
    \begin{equation} \label{drift}
        \mathbb{E} [\mathcal{L}(X_{n+1})-\mathcal{L}(X_{n})\mid X_n=x]<-1 +b \ind_C(x)
    \end{equation}
    Then the DTMP is positive Harris recurrent.
    
    Now choosing  $\mathcal{L}(x)=x$, we have:
    \begin{align*}
        &\mathbb{E} [\mathcal{L}(X_{n+1})- \mathcal{L}(X_{n})\mid X_n=x]
        =\mathbb{E} [X_{n+1}\mid X_n=x]-x \\
        &=\mathbb{E}  [ (x +Z(x)) \mathrm{e}^{-\gamma \Delta T_n(x+Z(x))})] -x \\
        &=x(\mathbb{E}  [ \mathrm{e}^{- \gamma\Delta T_n(x+Z(x))} ]-1) +\mathbb{E}[Z(x) \mathrm{e}^{-\gamma\Delta T_n(x+Z(x))})] 
    \end{align*}
    Let us start to consider the case $\theta+\phi<1$, then:
    
    \begin{align*}
    &\mathbb{E}[Z(x) \mathrm{e}^{- \gamma\Delta T_n(x+Z(x))})]<\mathbb{E}[Z(x) ] \\ &=p^{(n)}_I(x)\mathbb{E}[V(x)] + p^{(n)}_E (x) \mathbb{E}[W] \le (\beta+ \delta)  x^\theta
    \end{align*}
    for any constant $\delta>0$
    while 

    \begin{align*}
        &\mathbb{E}  [ \mathrm{e}^{- \gamma \Delta T_n(x+Z(x))} ]=\int \int  \mathrm{e}^{-\gamma \zeta} f_{\Delta t }(\zeta \mid x+v) f_Z(v\mid x)
        \mathrm{d}v\mathrm d \zeta \\
        &\le   \int \int  \mathrm{e}^{-\gamma \zeta} [\lambda_0+ \lambda_1(x+v)^\phi] e^{-[\lambda_0+ \lambda_1(x+v)^\phi]\zeta } 
        \mathrm d \zeta f_Z(v\mid x)
        \mathrm{d}v\\ 
       &= \int \frac{\lambda_0+ \lambda_1(x+v)^\phi }{\lambda_0+ \lambda_1(x+v)^\phi+\gamma} f_Z(v\mid x)
        \mathrm{d}v \\  
        &= 1- \int \frac{\gamma }{\lambda_0+ \lambda_1(x+v)^\phi+\gamma} f_Z(v\mid x)
        \mathrm{d}v\\
    \end{align*}
Note that the first inequality holds because the stochastic intensity in $(X_n,X_{n+1})$ 
is upper bounded by $ \lambda_0+ \lambda_1(X_n+Z(X_n))^\phi$.
Now  given an arbitrary $\alpha>1$
\begin{align*}
&\int_0^\infty \frac{\gamma }{\lambda_0+ \lambda_1(x+v)^\phi+\gamma} f_Z(v\mid x)
        \mathrm{d}v\\
        &=\Big(\int_0^{\alpha x} +\int_{\alpha x}^\infty\Big) \frac{\gamma }{\lambda_0+ \lambda_1(x+v)^\phi+\gamma} f_Z(v\mid x)
        \mathrm{d}v\\\ & \ge \frac{\gamma }{\lambda_0+ \lambda_1(\alpha x)^\phi+\gamma} F_Z\left((\alpha-1)x\mid x\right)
\end{align*}
    
    Therefore for $\alpha x>1$
\begin{align*}
&\mathbb{E}  [ \mathrm{e}^{- \gamma \Delta T_n(x+Z(x))} ]\le 1- \frac{\gamma }{\lambda_0+ \lambda_1(\alpha x)^\phi+\gamma} F_Z\left((\alpha-1)x\mid x\right)\\
&\le 1- \frac{\gamma }{\lambda_0+ \lambda_1+\gamma}(\alpha x)^{-\phi} F_Z\left((\alpha-1)x\mid x\right)= 1- c x^{-\phi} 
\end{align*}
with $c:=\frac{\gamma }{\lambda_0+ \lambda_1+\gamma}\alpha^{-\phi} F_Z\left((\alpha-1)x\mid x\right)$ and  
    \begin{align}
        \Delta \mathcal{L}(x):=\mathbb{E} [\mathcal{L}(X_{n+1})- \mathcal{L}(X_{n})\mid X_n=x]\nonumber\\
        < x(1 -  cx^{-\phi} -1 ) + (\beta+\delta)x^\theta\nonumber\\
        = -c x^{1-\phi}+ (\beta+\delta) x^\theta  \label{drift-components}
    \end{align}
    Now  if $\theta<1-\phi$, the first negative term asymptotically for large $x$ dominates,
      when instead $\theta+\phi=1$  the two terms in \eqref{drift-components} 
      are comparable and    a negative drift is observed  only under the additional condition that $\frac{ c}{\beta} >1$.
    In both cases   there exists a $x_0$ such that  $\Delta \mathcal{L}(x)< -1  $ for any $x\ge x_0$. 
    Moreover 
    \[
    \Delta \mathcal{L}(x)=\mathbb{E}  [ (x +Z(x)) \mathrm{e}^{- \Delta T(x+Z(x))})] -x
    \le \mathbb{E}[Z(x)]
    \]
    therefore 
    \[
    \sup_{x\le x_0} \Delta \mathcal{L}(x)=  \sup_{x\le x_0}  \mathbb{E}[Z(x)]\le  \mathbb{E}[Z(x_0)]. 
    \]
    In conclusion condition \eqref{drift} is met by setting $C=\{x\le x_0\}$ and $b=\mathbb{E}[Z(x_0)] +1$.
    
\end{proof}
    
    At last, with rather standard arguments it can be shown that the ergodicity of $\{X_n\} _n$ implies  the ergodicity of $\{X(t)\}$
    (i.e. a unique stationary distribution $\Pi(A)$  exists and   $P^t(x,A)\to \Pi(A)$  for every $x\in \mathbb{R}^+$  and $A\in \mathcal{B}(\mathbb{R}^+)$.
    
    Now, since  the measure induced by $P^t(x,A)$ can be easily shown to be absolutely continuous with respect to the Lebesgue measure $\mu_{\text{Leb}}(\mathbb{R})^+)$,
    for any $t>0$
    we have that: 
    \[
    \Pi(t,A)=\int_{y\in A}\int_{x\in \mathrm{R}^+} P^t(x, \mathrm{d}y) \Pi(0,\mathrm{d}x)  \ \
    \]
    is, as well, absolutely continuous  with respect to the Lebesgue measure $\mu_{\text{Leb}}(\mathbb{R}^+)$ with density  (i.e Radon-Nikodyn derivative)
    $f(x,t)$.

    \subsection{Kolmogorov equation}
    
    Now defined with $F(y,t):= \mathbb{P}(X(t)\le y )=\int_0^x f(x,t) \mathrm{d}x $   we have that th $F(y,t)$ satisfies the following Partial Integro-Differential Equation:
    \begin{equation}\label{diff-eq-5.1}
        \frac{\partial F(y,t)}{\partial t}=  \gamma y \frac{\partial F(y,t)}{\partial y}- \int  ( \lambda(t)+\mu)\bar F_Z( y-x \mid x) \frac{\partial F(x,t) }{\partial x}     \mathrm{d}x
    \end{equation}
    While the unique stationary distribution $F(y)$ satisfies the following ODE:
    \begin{equation}\label{stat-diff-eq-5.1}
        \gamma y \frac{\mathrm{d} F(y)}{\mathrm{d} y}=\int  (\lambda(t)+\mu)\bar F_Z( y-x \mid x) \frac{\mathrm{d}  F(x) }{\mathrm{d}  x}     \mathrm{d}x
    \end{equation}
    
\begin{proof}
    Observe that  by construction:
    \[
    F(y, t +\Delta t ) = F(y, t ) + \Delta F^+(y,t)- \Delta F^-(y,t)
    \]
    with 
    \[
    \Delta F^+(y,t)=\mathbb{P} ( X(t+\Delta t)\le y,   X(t)> y )
    \]
    and
    \[
    \Delta F^-(y,t)=\mathbb{P} ( X(t+\Delta t)>y,   X(t)\le y )
    \]
    now  denoted with $N(t,\Delta t):= N(t+\Delta t)-N(t)$, the number of jumps in $[t, t+ \Delta t)$, as immediate consequence of the definition of stochastinc intensity we have
    that:
    
    \begin{align*}
        &\mathbb{P}(N(t,\Delta t)=0\mid X(t)=y )= 1- (\lambda(y)+\mu)\Delta t + o(\Delta t)\\
        &\mathbb{P}(N(t,\Delta t)=1\mid X(t)=y  )= (\lambda (y)+\mu) \Delta t + o(\Delta t)\\
        &\mathbb{P}(N(t,\Delta t)>1\mid X(t)=y )= o(\Delta t)
    \end{align*}
    where $\lambda(y)=\lambda^{(0)}+\lambda^{(1)}y^\phi$.
    
    Now:

    \begin{align*}
        &\Delta F^+(y,t)=\sum_{k\in \mathbb{N}\cup \{0\} } \mathbb{P} ( X(t+\Delta t)\le y,   X(t)> y, N(t,\Delta t)=k )\\
    \end{align*}

    Now
    \begin{align*}
        &\mathbb{P} ( X(t+\Delta t)\le y,   X(t)> y,  N(t,\Delta t)=0)\\
        =&\mathbb{P} ( X(t) \mathrm{e}^{-\gamma \Delta t} \le y,   X(t)> y, N(t,\Delta t)=0 )\\
        =&\mathbb{P} ( y<  X(t) \le y \mathrm{e}^{\gamma \Delta t }, N(t,\Delta t)=0  ) \\
        =&\mathbb{P}( N(t,\Delta t)=0 \mid   y<  X(t) \le y \mathrm{e}^{\gamma \Delta t } )
         \mathbb{P} ( y<  X(t) \le y \mathrm{e}^{\gamma \Delta t } ).
    \end{align*}
    with
    \begin{align*}
        &\mathbb{P} ( y<  X(t) \le y \mathrm{e}^{\gamma \Delta t } ) 
        =\int_{y}^{ y \mathrm{e}^{\gamma \Delta t} }  f(z, t)\mathrm{d}z\\
        &=\int_{y}^{ y (1+\gamma y \Delta t+ o(\Delta t)) }  f(z, t)\mathrm{d}z
        =   (\gamma y \Delta t +o(\Delta t)  ) f(y,t) \\
        & = (\gamma y \Delta t +o(\Delta t)  ) \frac{F(y,t)}{\partial y} 
    \end{align*}
    
    and 
    
    \[
    \mathbb{P}( N(t,\Delta t)=0 \mid   y<  X(t) \le y \mathrm{e}^{\gamma \Delta t } )=1- \lambda(y)\Delta t +o(\Delta t)
    \]
    by the continuity of $\lambda(y )$ with respect to $y$.  
    Moreover  by construction
    \begin{align*}
        &\{\ X(t+\Delta t)\le y,   X(t)> y, N(t,\Delta t)=1\} \\
        &\subseteq \{   y<  X(t) \le y \mathrm{e}^{\gamma \Delta t} ,   N(t,\Delta t)=1 \}
    \end{align*}
    therefore 
    \begin{align*}
        \mathbb{P}(  y<  X(t) \le y \mathrm{e}^{\gamma \Delta t} ,   N(\Delta t)=1  )= o(\Delta t)
    \end{align*}
    
    Similarly
    \[
    \Delta F^-(y,t)= \sum_{k\in \mathbb{N}\cup \{0\} } \mathbb{P} ( X(t+\Delta t)> y,   X(t)< y, N(t,\Delta t)=k )
    \]
    Now observe that  the event 
    \[
    \{   X(t+\Delta t)> y,   X(t)< y, N(t, \Delta t)=0    \}=\emptyset
    \]
    
    Therefore 
    \[
    \Delta F^-(y,t)= \mathbb{P} ( X(t+\Delta t)> y,   X(t)< y, N(t,\Delta t)=1) +o(\Delta t )
    \]
    since
    \[
    \mathbb{P}(    N(t,\Delta t)>1 \mid X(t)<y  )    = o(\Delta t)
    \]
    Now
    \begin{align*}
        &\mathbb{P} ( X(t+\Delta t)> y,   X(t)< y, N(t,\Delta t)=1)\\
        &= \int_{x<y} \mathbb{P} ( X(t+\Delta t)> y,  N(t,\Delta t)=1 \mid  X(t)=x) \mathrm{d}F(x,t)   \\
        &= \int_{x<y} \mathbb{P} ( X(t+\Delta t)> y  \mid N(t,\Delta t)=1,  X(t)=x)\\
        &\qquad  \mathbb{P} (   N(t,\Delta t)=1 \mid  X(t)=x)  \mathrm{d}F(x,t)  \\
        &= \Delta t \int_{x<y} \bar F_Z(y-x\mid x)(\lambda(x)+\mu)\mathrm{d}F(x,t)  +o(\Delta t )
    \end{align*}
\end{proof}

    Then, summarizing:
    \begin{align*}
        &\frac{F(y,t+\Delta t)-F(y,t)}{\Delta t }=(\gamma y  +o(1)  ) \frac{F(y,t)}{\partial y} \\ 
        &\quad -\int_{x<y} \bar F_Z(y-x\mid x)(\lambda(x)+\mu)\mathrm{d}F(x,t)  + o(1)
    \end{align*}
    The result  is then obtained taking  liminf and limsup of both sides for $\Delta t\to 0$.
    
    Equation  \eqref{stat-diff-eq-5.1} is then immediately obtained from  \eqref{diff-eq} 
    by  looking for  solutions that are stationary  i.e., $F(y,t)=F(y)$.

\end{document}